\documentclass[journal]{article}
\usepackage{amsmath,amsfonts,stmaryrd,amssymb}
\usepackage{extarrows}
\usepackage{bm}
\usepackage{paralist}
\usepackage{float}
\usepackage{url}
\usepackage{booktabs} 

\usepackage{cryptocode}

\usepackage{enumitem}

\usepackage[normalem]{ulem}
\usepackage{hyperref}
\usepackage{multirow}

\usepackage{tabularx}
\usepackage{pgf}
\usepackage{tikz}
\usetikzlibrary{arrows,automata}
\newcommand{\Z}{\mathbb{Z}}

\newcommand{\N}{\mathbb{N}}

\newtheorem{theorem}{Theorem}      
\newtheorem{proof}{Proof}

\newtheorem{rem}{Remark}[subsection]
\newtheorem{ex}{Example}[subsection]

\def\BibTeX{{\rm B\kern-.05em{\sc i\kern-.025em b}\kern-.08em
    T\kern-.1667em\lower.7ex\hbox{E}\kern-.125emX}}
    
\makeatletter
\newcommand{\linebreakand}{%
  \end{@IEEEauthorhalign}
  \hfill\mbox{}\par
  \mbox{}\hfill\begin{@IEEEauthorhalign}
}
\makeatother

\usepackage{subcaption}
\usepackage{graphicx}
\DeclareCaptionLabelFormat{custom}
{%
      \textbf{#1 #2}
}
\DeclareCaptionLabelSeparator{custom}{--}
\DeclareCaptionFormat{custom}
{%
    #1#2\small #3
}
\usepackage[a4paper, total={7in, 10in}]{geometry}

\begin{document}

\title{Exploit the Leak: Understanding Risks in Biometric Matchers}

\author{DURBET Axel$^1$,\\
axel.durbet@uca.fr
\and
THIRY-ATIGHEHCHI Kevin$^1$,\\ 
kevin.atighehchi@uca.fr
\and
CHAGNON Dorine$^1$,\\
dorine.chagnon@uca.fr
\and
GROLLEMUND Paul-Marie$^2$\\
paul\_marie.grollemund@uca.fr}

\date{%
$^1$LIMOS, Université Clermont-Auvergne, CNRS, Mines Saint-Étienne, Aubière, France\\
$^2$LMBP, Université Clermont-Auvergne, CNRS, Mines Saint-Étienne, Aubière, France\\}

\maketitle

\begin{abstract}
    In a biometric authentication or identification system, the matcher compares a stored and a fresh template to determine whether there is a match. This assessment is based on both a similarity score and a predefined threshold.
    For better compliance with privacy legislation, the matcher can be built upon a privacy-preserving distance.
    Beyond the binary output (`\textit{yes}' or `\textit{no}'), most schemes may perform more precise computations, \textit{e.g.}, the value of the distance.
    Such precise information is prone to leakage even when not returned by the system.
    This can occur due to a malware infection or the use of a weakly privacy-preserving distance, exemplified by side channel attacks or partially obfuscated designs.
    This paper provides an analysis of information leakage during distance evaluation.
    We provide a catalog of information leakage scenarios with their impacts on data privacy.
    Each scenario gives rise to unique attacks with impacts quantified in terms of computational costs, thereby providing a better understanding of the security level.
\end{abstract}


\section{Introduction}\label{intro}
Biometric authentication protocols involve the comparison of a fresh biometric template with the reference template. This comparison computes the distance between the newly acquired data and the stored template. If this distance is below a given threshold, access is granted; otherwise, it is denied. 
While many protocols use standard metrics such as the Hamming distance~\cite{bernal2023review}, this process can inadvertently leak information that adversaries might exploit to reconstruct the stored template. Vulnerabilities arise from implementation errors, inherent flaws, and server-level attacks such as malware~\cite{sharma2023survey}, which can compromise system-wide security.
Furthermore, Aydin and Aysu~\cite{HomoLeak} and Hashemi~\textit{et al.}~\cite{hashemi2023time} have highlighted an increasing prevalence of side-channel attacks. 
Side-channel techniques, including timing, differential power analysis, cache-based, electromagnetic, acoustic, and thermal attacks, exploit various operational artifacts to extract sensitive information~\cite{sharma2023survey}.
One of the concerns is the partial or total leakage of distance computation information.
Such information leakage poses significant security and privacy risks, especially in sensitive applications like privacy-preserving applications (\textit{e.g.}, biometric recognition systems). 
In this paper, we focus on the following attacks:
\begin{itemize}
\item \textit{Offline exhaustive search attacks} refer to scenarios for which a leaked yet obfuscated database is available for an attacker. The attacker employs the public transformation to verify a candidate vector. This verification may give additional information beyond the minimal information leakage (\textit{`yes'} or \textit{`no'}), for example via side-channel attacks.
\item \textit{Online exhaustive search attacks} correspond to attacks for which an attacker must interact with the biometric system to infer information about the targeted vector. Then, the attacker needs to force the system to leak additional information beyond the minimal information leakage (\textit{`yes'} or \textit{`no'}), for example via a malware infection.
\end{itemize}

\paragraph*{Related Works}
To the best of our knowledge, two papers investigate information leakage of biometric systems using privacy-preserving distance.
Pagnin~\textit{et al.}~\cite{pagnin2014leakage} shows that the output of a privacy-preserving distance can be exploited to infer the hidden input.
This type of attack is considered the most devastating for such systems, as evidenced by Simoens~\textit{et al.}~\cite{Simoens2012AFF}.
The work of Pagnin~\textit{et al.} takes place in the minimal leakage scenario, wherein only the binary output of the biometric system is given to the attacker.
The authors present the \textit{Center Search Attack}, designed to recover the hidden enrolled input for any `valid' biometric template in $\mathbb{Z}_q^n$, where `valid' refers to inputs within a ball centered at the enrolled template and with a radius equal to the decision threshold $t$. 
To efficiently locate a valid input, the authors also examine the exhaustive search attack, particularly its application on binary templates ($q=2$). 
They suggest implementing a \textit{sampling without replacement} strategy using their \textit{Tree algorithm} to streamline the identification of a suitable input for the Center Search Attack.
This efficient identification of a proper input requires a number of authentication attempts that is exponential in the space dimension $n$ minus the threshold $t$.
While their work focuses on the minimal leakage scenario, our analysis includes the consideration of multiple additional information leaks that may arise during the matching operation.

\paragraph*{Contributions}
We analyze the impact of potential information leakage in distance evaluations.
Our contributions detail various leakage scenarios, their corresponding generic attacks, and the computational costs involved:
\begin{itemize}
\item 
We revisit the exhaustive search attack in the scenario of a minimal (one-bit) information leakage, correcting a previously cited result (see~\cite{pagnin2014leakage}) about the costs of optimal and near-optimal strategies.
\item 
We introduce new attack strategies by malicious clients that exploit various levels of non-minimal information leaks from the system. Our complexity results, which detail the cost of these attacks, apply to both \textit{offline exhaustive search attacks} that leverage a leaked (yet obfuscated) database and \textit{online exhaustive search attacks} involving direct interactions with the server.
\item 
We investigate a novel attack, named accumulation attack, where an \textit{honest-but-curious} server accumulates knowledge during client authentication. This type of attack occurs when there is a minor, yet non-negligible, amount of information leakage.
\end{itemize}

The complexities of the attacks, relying on different scenarios, are summarized in Figure~\ref{Leak_table}. 
\begin{table*}
\resizebox{\textwidth}{!}{
\centering
\begin{tabular}{ccccc}
\toprule
\multicolumn{1}{l}{Distance-to-Threshold comparison}   & Leakage                      & Complexity type & Complexity in Big-Oh & Theorem\\
\midrule
\multirow{4}{*}{Below} & Distance                     & Exponential       & $q^{n-\varepsilon}+q\varepsilon$  & \ref{th-below-leak-dist}\\
                               & Positions                    & Exponential       & $q^{n-\varepsilon}+q$  & \ref{th-below-leak-pos}\\
                               & Positions and values & Exponential       & $q^{n-\varepsilon}$               & \ref{th-below-leak-value-and-pos}\\
                               & Positions and values (accumulation) & Linearithmic/Polynomial & $n^{\alpha}\log{n}$ & \ref{th-acc}\\
                               \midrule
\multirow{4}{*}{Both} & Minimal$^{1}$                      & Exponential       & $2^{n-\varepsilon}+n+2\varepsilon$  & \ref{th-both-minimal}\\

                               & Distance                     & Linear       & $nq$ & \ref{th-both-dist-leak}\\
                               & Positions                     & Constant       &  $q$ & \ref{th-both-pos-leak}\\
                               & Positions and values & Constant    & $1$  & \ref{th-both-pos-and-val-leak}       \\

                               \bottomrule
\end{tabular}
}
\caption[]{Summary of all leakage exploits and their complexities with $\alpha$ such that the occurrence of the rarest error is $n^{-\alpha}$ with $\alpha\in\mathbb{R}_{\ge 1}$. 
The Distance-to-Threshold comparison determines if the leak occurs when $d(x,y)\leq \varepsilon$ (below) or when there is no distance requirement between $x$ and $y$ (both).
For all the complexities, $x$ and $y$ are in $\Z_q^n$ with $q\ge2$ except for the minimal leakage where $x$ and $y$ are in $\Z_2^n$. The provided complexities represent worst-case scenarios, except for the accumulation attack where the result is the expectation.
\newline$^{1}$Note that the Big-Oh complexity of the optimal exhaustive search strategy, in the worst-case, is the same as the naive strategy as the minimum of $h(\cdot)$ is $0$.
}
\label{Leak_table}
\end{table*}

\paragraph*{Outline}
Section~\ref{prel} introduces notations and terminologies and classifies the different types of information leakages.
Section~\ref{attacks} begins by revisiting the exhaustive search attack in the minimal (one-bit) information leakage scenario, including a correction of a previously cited result concerning the costs of optimal and near-optimal strategies. It then introduces new strategies for attacks by malicious clients capturing various other types of information leakages, covering both offline and online exhaustive search attacks, with an emphasis on their computational costs. The section concludes by examining accumulation attacks performed by an "honest-but-curious" server during client authentication, detailing the computational cost involved.
Section~\ref{Conclu} provides a discussion of the presented results.

\section{Preliminaries}
\label{prel}
This section introduces the notations as well as the attacker model and, a list of the considered information leakage scenarios.

\subsection{Notations and Attacker Models}
\label{Bg}\label{Notation}
Let $\Z_q^n = \lbrace 0,\dots,q-1 \rbrace^n$ be a metric space equipped with the Hamming distance $d$ and $\varepsilon\in\N$ a threshold.
The Hamming distance is defined by
$$d(x, y) = \left|\lbrace i\in\lbrace 1,\dots n,\rbrace | x_i \not = y_i\rbrace\right|$$
for two vectors $x=(x_1,\dots,x_n)$ and $y=(y_1,\dots,y_n)$ in $\Z_q^n$.
Let $\texttt{Match}_{x,\varepsilon}$ denote the oracle modeling the interaction between the biometric system using a privacy-preserving distance and the attacker.
$\texttt{Match}_{x,\varepsilon}$ receives the template selected by the attacker and compares it with the previously enrolled and stored template. If the distance is below the threshold $\varepsilon$, the oracle returns $1$ and $0$ otherwise. In a more formal way, $\texttt{Match}_{x,\varepsilon}$ is a function defined as:
\begin{eqnarray*}
    \texttt{Match}_{x,\varepsilon}:\mathbb{Z}_q^n & \longrightarrow & \lbrace0,1\rbrace\\
    y & \longmapsto & \begin{cases}
        1 \text{ if } d(x,y)\leq \varepsilon.\\
        0 \text{ otherwise.}
    \end{cases}
\end{eqnarray*}

A privacy-preserving distance may leak additional information beyond its binary output. Under the specifications of each scenario, the oracle may display this additional information.
The objective of the attacker is to find the hidden template $x$ exploiting the oracle outputs. In the context of a biometric system, the objective of the attacker may be relaxed to simply find $y$ that is close to $x$ with respect to $d$ and $\varepsilon$.

\subsection{Typology of Information Leakage}
\label{Leak_type}
\noindent
In the context of a biometric system, a critical vulnerability arises when information is intercepted between the matcher and the decision module, as illustrated in Figure~\ref{fig:notratha} (point $8$). This figure, inspired by Ratha~\textit{et al.}~\cite{Rathasheme}, provides an overview of the attack points in biometric systems while introducing both the decision module and two additional attack points. 
Except for the accumulation attack, the attacker exploits points $4$ and $8$ in all discussed scenarios. Point $4$ allows the submission of a chosen template, while point $8$ grants access to additional information beyond the binary output. The accumulation attack only necessitates control over the point $8$. For detailed insights into the remaining attack points, readers are referred to Ratha~\textit{et al.}~\cite{Rathasheme}.
There are three main categories of information leakage: 
\begin{itemize}
    \item Below the threshold.
    \item Above the threshold. 
    \item Both below and above the threshold.
\end{itemize} 

In each of these categories, several sub-settings can be identified.
The first one corresponds to the absence of any leakage, resulting in $\texttt{Match}_{x,\varepsilon}$ yielding only the binary output. 
Then, the following information leakages are examined: 
\begin{itemize}
    \item The distance. 
    \item The positions of the errors.
    \item Both the error positions and values.
    \item Both the distance and the positions of the errors.
    \item Both the distance and the positions and their corresponding erroneous values.
\end{itemize} 
It is not relevant to consider that additional information is leaked only above the threshold, as no scheme has such behavior. As a consequence, solely scenarios `below the threshold' and `below and above the threshold' are examined.
The Hamming distance is a measure of the number of differing coordinates between two templates. Therefore, knowledge of the erroneous coordinates implies knowledge of the distance itself. Hence, we do not consider all possible scenarios.

\begin{figure*}
    \centering
    \includegraphics[width=\textwidth]{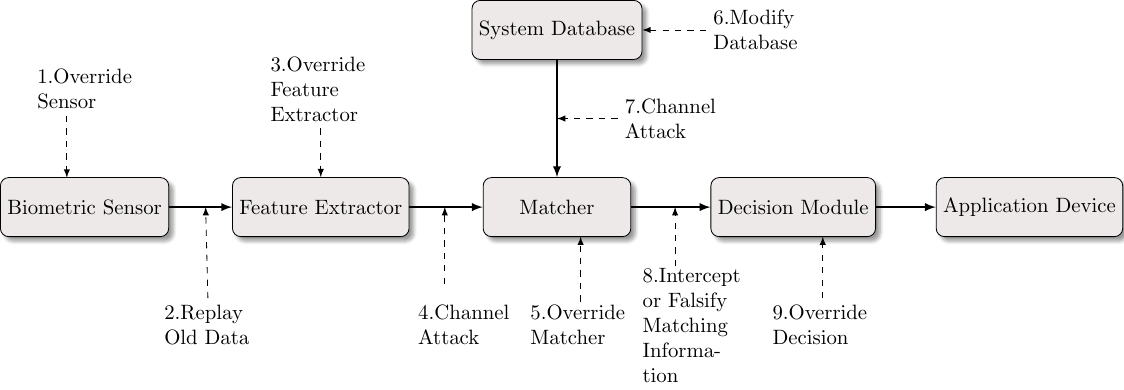}
    \caption{Attack points in a generic biometric recognition system.}
    \label{fig:notratha}
\end{figure*}

\section{Exploiting the Leakage}
\label{attacks}
\subsection{Active Attacks}

\subsubsection{Attack Complexity for the Minimal (One-bit) Leakage}

The exhaustive search may be enhanced for both computationally bounded and unbounded adversaries.
The objective is to identify the minimal number of balls of radius $\varepsilon$ to partition the space and then, to query the center of each ball in order to remove an entire ball if the query fails. This is an instance of the set covering problem.
Pagnin~\textit{et al.}~\cite{pagnin2014leakage} proposed to use such a method and claimed that the number of points that the adversary needs to query is only a factor of $\mathcal{O}(\varepsilon \ln(n+1))$ more than the optimal cover. However, the result is imprecise, as detailed below, mainly because the optimal cover is not given.

\paragraph*{Unbounded Case}
The adversary solves the NP-hard set covering problem~\cite{korte2011combinatorial} to find the optimal covering of $\Z_q^n$ using balls of radius $\varepsilon$.
The adversary exhaustively searches $x$ using at most $q^{n(1-h_q(\varepsilon/n))+o(n)}$ queries to $\texttt{Match}_{x,\varepsilon}$.
The number of vectors involved in a given optimal cover is $\frac{q^n}{|B_{q,\varepsilon}(x)|}$, which can be asymptotically approximated as detailed in what follows.  
First, recall that the cardinal of $B_{q,\varepsilon}(x)$ is $$|B_{q,\varepsilon}(x)| = \sum_{i=0}^\varepsilon \binom{n}{i}(q-1)^i,$$ and that the $q$-ary entropy is $h_q(x) = x \log_q(q - 1) - x \log_q x - (1 - x) \log_q(1 - x).$
Then, using bounds on the binomial coefficient (see~\cite{thomas2006elements}), the result follows if $\frac{\varepsilon}{n} \leq 1-\frac{1}{q}$ holds and if $n$ is large enough.

\paragraph*{Bounded Case}
The adversary may use a greedy algorithm to find a non-optimal covering containing $\frac{q^nH(n)}{|B_{q,\varepsilon}|}$ vectors~\cite{chvatal1979greedy} with $H(n)=\sum_{i=1}^n i^{-1}$ the $n$-th harmonic number.
The adversary then finds a solution with an exhaustive search in at most $\frac{q^nH(n)}{|B_{q,\varepsilon}|}$ queries.
In order to provide a more intuitive value, notice that $\frac{q^nH(n)}{|B_{q,\varepsilon}|}$ can be bounded up by $\frac{q^n(\ln(n)+1)}{|B_{q,\varepsilon}|}$.
As in the unbounded case, using the $q$-ary entropy and Stirling's approximation, this non-optimal covering leads the attacker to make at most $q^{n(1-h_q(\varepsilon/n))+o(n)}$ queries, as $\log_q(\ln(n)+1)=o(n)$.

\subsubsection{Attack Complexities for Leakage Below the Threshold}
\label{below}
Leakage below the threshold is considered in this section. 
Given the hidden target $x$, querying $y$ such that $d(x,y)\leq\varepsilon$ to the oracle $\texttt{Match}_{x,\varepsilon}$ provides information beyond the binary output.

\paragraph{Leakage of the Distance}
The first case occurs when the distance is given to the attacker as extra information. 
\begin{theorem}\label{th-below-leak-dist}
    Given $\varepsilon$ a threshold, $x\in\mathbb{Z}_q^n$ a vector, and $\texttt{Match}_{x,\varepsilon}$ leaks the distance below the threshold, an attacker can retrieve $x$ in the worst case in $\mathcal{O}(q^{n-\varepsilon}+q\varepsilon)$ queries to $\texttt{Match}_{x,\varepsilon}$.
\end{theorem}

\begin{proof}
The system, using the Hamming distance, requires a minimum of $n-\varepsilon$ accurate coordinates to output $0$.
Since the attacker specifically targets $n-\varepsilon$ coordinates (the attacker arbitrarily chooses $\varepsilon$ coordinates that do not change), an exhaustive search attack is performed in at most $q^{n-\varepsilon}$ steps to get accepted by the system.
Then, a hill-climbing attack runs on the remaining $\varepsilon$ coordinates to minimize the distance at each step. Coordinate by coordinate, the attacker obtains the right value if the distance decreases. Since there are $q$ different values to test on $\varepsilon$ coordinates, determining the correct ones requires a maximum of $(q-1)\varepsilon$ steps. Then, the overall complexity is $\mathcal{O}(q^{n-\varepsilon}+q\varepsilon)$.\hfill $\blacksquare$
\end{proof}

\paragraph{Leakage of the Positions}
The positions of the errors are the extra information given to the attacker, while their values remain secret. 
\begin{theorem}\label{th-below-leak-pos}
    Given $\varepsilon$ a threshold, $x\in\mathbb{Z}_q^n$ a vector, and $\texttt{Match}_{x,\varepsilon}$ leaks the positions of the errors below the threshold, an attacker can retrieve $x$ in the worst case in $\mathcal{O}(q^{n-\varepsilon}+q)$ queries to~$\texttt{Match}_{x,\varepsilon}$.       
\end{theorem}

\begin{proof}
As the leakage occurs solely below the threshold, the first step is to find a vector $y\in\mathbb{Z}_q^n$ such that $d(x,y)\leq \varepsilon$. To identify such a vector, the attacker performs an exhaustive search attack in $q^{n-\varepsilon}$ steps, as previously shown. Since $\varepsilon$ coordinates remain unknown, and each coordinate ranges from $0$ to $q-1$, every possibility must be examined. By testing all possibilities simultaneously -- for instance, testing all coordinates at $0$, then all coordinates at $1$, and so forth up to $q-2$ while retaining the correct values -- the original vector can be identified in no more than $q-1$ queries (refer to the example illustrated in Figure~\ref{fig:reverse_master_mind}). Therefore, the complexity of the attack for recovering $x$ is $\mathcal{O}(q^{n-\varepsilon}+q)$.
\hfill $\blacksquare$
\end{proof}

Figure~\ref{fig:reverse_master_mind} gives a representation of the attack described above in the case $\mathbb{Z}_4^5$ and the hidden vector or the missing coordinates is $(0,1,3,2,2)$. Note that the actual complexity is $q-1$ since the final exchange is unnecessary, as the coordinates at $q-1$ become known after $q-1$ queries by inference.

\paragraph{Leakage of the Positions and the Values}
When a vector below the threshold is given to the oracle $\texttt{Match}_{x,\varepsilon}$, the attacker gets information about both error positions and their values. This is similar to an error-correction mechanism designed to correct errors below a given threshold. Note that in the binary case, this scenario is the same as the previous one, hence the only considered case is $q>2$.
\begin{theorem}\label{th-below-leak-value-and-pos}
    Given $\varepsilon$ a threshold, $x\in\mathbb{Z}_q^n$ a vector, and $\texttt{Match}_{x,\varepsilon}$ leaks the positions and the values of the errors below the threshold, an attacker can retrieve $x$ in $\mathcal{O}(q^{n-\varepsilon})$ queries to $\texttt{Match}_{x,\varepsilon}$.
\end{theorem}

\begin{proof}
First, an exhaustive search is performed to find a vector $y$ for which the distance is below the threshold, for a cost of $\mathcal{O}(q^{n-\varepsilon})$. Then, given the error positions and the corresponding error values, $y$ yields immediately the recovery of $x$. 
In the end, the complexity of the attack is $\mathcal{O}(q^{n-\varepsilon})$.
\hfill $\blacksquare$
\end{proof}

\subsubsection{Leakage Below and Above the Threshold}
\label{both}
The second scenario is considered in this section, which involves a leakage independent of the threshold. In other words, when a hidden vector $x$ is targeted, the queried vector $y$ to the oracle $\texttt{Match}_{x,\varepsilon}$ results in the leak of additional information.

\paragraph{Minimal Leakage (single bit of information leakage)}
The basic usage of the system is characterized by the minimal leakage scenario, where the binary output itself is considered a necessary leakage. This minimal leakage is indispensable for the system's work and is consistent across these scenarios as the system always responds. Remark that if the server does not answer above the threshold, the non-answer gives the attacker the wanted information.
\begin{theorem}\label{th-both-minimal}
    Given $\varepsilon$ a threshold, $x\in\mathbb{Z}_2^n$ a vector, and $\texttt{Match}_{x,\varepsilon}$ that does not leak any extra information, an attacker can retrieve $x$ in $\mathcal{O}(2^{n-\varepsilon}+n+2\varepsilon)$ queries to $\texttt{Match}_{x,\varepsilon}$.
\end{theorem}

\begin{proof}
As in the previous cases, the attacker seeks a vector $y$ below the threshold. Such a vector is found by exhaustive search in $2^{n-\varepsilon}$ steps. Then, the attacker performs the center search attack~\cite{pagnin2014leakage} to retrieve the original data in at most $n+2\varepsilon$ queries. The complexity of the attack to find $x$ is $\mathcal{O}(2^{n-\varepsilon}+n+2\varepsilon)$.\hfill $\blacksquare$
\end{proof}

\paragraph{Leakage of the Distance}
In this case, $d(x,y)$ the distance between $y\in\mathbb{Z}_q^n$ the fresh template and $x\in\mathbb{Z}_q^n$ the old template is leaked to the attacker regardless of the threshold.
\begin{theorem}\label{th-both-dist-leak}
    Given $\varepsilon$ a threshold, $x\in\mathbb{Z}_q^n$ a vector, and $\texttt{Match}_{x,\varepsilon}$ leaks the distance, an attacker can retrieve $x$ in $\mathcal{O}(nq)$ queries to $\texttt{Match}_{x,\varepsilon}$.       
\end{theorem}

\begin{proof}
    As the attacker has access to the distance, it is possible to perform a hill-climbing attack, trying to minimize the distance at each step. The strategy is to find the vector $y$, coordinate by coordinate. As each coordinate has $q$ possible values and there are $n$ coordinates, this is done in $\mathcal{O}(nq)$ steps. 
\hfill $\blacksquare$
\end{proof}

\paragraph{Leakage of the Positions}
The extra information given to the attacker is the positions of the errors.
\begin{theorem}\label{th-both-pos-leak}
    Given $\varepsilon$ a threshold, $x\in\mathbb{Z}_q^n$ a vector, and $\texttt{Match}_{x,\varepsilon}$ leaks the positions of the errors, an attacker can retrieve $x$ in $\mathcal{O}(q)$ queries to $\texttt{Match}_{x,\varepsilon}$.       
\end{theorem}

\begin{proof}
She tries the vector $(0,\ldots,0)$, $(1,\ldots,1)$ up to, $(q-1,\ldots,q-1)$ and keep for each coordinate the right value (see Figure~\ref{fig:reverse_master_mind}). Hence, the complexity of the attack to recover $x$ is $\mathcal{O}(q)$.\hfill $\blacksquare$
\end{proof}

\begin{figure}
    \centering
    \addtolength{\tabcolsep}{-0.14cm}
    \begin{tabular}{cccccc}
        \uline{Queries:} &  $(\boxed{0},$ & $\underset{\times}{0},$ & $\underset{\times}{0},$ & $\underset{\times}{0},$ & $\underset{\times}{0})$ \\
                         &  $(\underset{\times}{1},$ & $\boxed{1},$ & $\underset{\times}{1},$ & $\underset{\times}{1},$ & $\underset{\times}{1})$ \\
                         &  $(\underset{\times}{2},$ & $\underset{\times}{2},$ & $\underset{\times}{2},$ & $\boxed{2},$ & $\boxed{2})$ \\
                         &  $(\underset{\times}{3},$ & $\underset{\times}{3},$ & $\boxed{3},$ & $\underset{\times}{3},$ & $\underset{\times}{3})$ \\
        \uline{Solution:} &  $(0,$ & $1,$ & $3,$ & $2,$ & $2)$ \\
    \end{tabular}
    \caption{Exploiting the error position leaked in the case $\mathbb{Z}_4^5$ and the hidden vector or missing coordinates is $(0,1,3,2,2)$.}
    \label{fig:reverse_master_mind}
\end{figure}
\addtolength{\tabcolsep}{0cm}

\paragraph{Leakage of the Positions and the Values}
In this last case, the positions of the errors and corresponding values are leaked. Unlike the scenario of leakage below the threshold, such a leak provides an error-correcting code mechanism that operates irrespective of any distance and threshold. \begin{theorem}\label{th-both-pos-and-val-leak}
    Given $\varepsilon$ a threshold, $x\in\mathbb{Z}_q^n$ a vector, and $\texttt{Match}_{x,\varepsilon}$ leaks the positions of the errors and their values, an attacker can retrieve $x$ in $\mathcal{O}(1)$ queries to $\texttt{Match}_{x,\varepsilon}$.
\end{theorem}

\begin{proof}
The submission of any vector gives the position of each error, and how to correct them, yielding a complexity in $\mathcal{O}(1)$.\hfill $\blacksquare$
\end{proof}

\subsection{Accumulation Attack: A Passive Attack}
\label{AccuAtt:section}
During the client authentications, the attacker passively gathers information by observing errors leaked by the server.
More specifically, the server leaks a list of positions and errors computed over the integers (\textit{i.e.}, $x_i - y_i$) made by a genuine client during each authentication. Such information gathered during one successful authentication attempt is called an observation.
The attacker aims to partially or fully reconstruct $x$ by exploiting these observations.

In the binary case (\textit{i.e.}, $q=2$), the errors precisely yield the bits. If $x_i - y_i = 1$ then $x_i=1$, and if $x_i - y_i = -1$ then $x_i=0$.
This attack is related to the Coupon Collector's problem~\cite{CoupnCollector}, which involves determining the expected number of rounds required to collect a complete set of distinct coupons, with one coupon obtained at each round, and each coupon acquired with equal probability. 
\begin{ex}
    \label{ExampleAccuAtt}
    Suppose a setting with a metric space $\mathbb{Z}_2^n$ equipped with the Hamming distance. A client seeks to authenticate to an \textit{honest-but-curious} server that uses a scheme leaking $d(x,y)$ and the corresponding errors if $d(x,y) \leq \varepsilon$.
    As the client is legitimate, \textit{i.e.}, $d(x,y) \leq \varepsilon$ with a high probability, the attacker recovers the values of at most $\varepsilon$ erroneous bits. The attacker needs to collect all the bits of the client, turning this problem into a Coupon Collector problem.
    For example, let assume $x = (0,0,1,1,0,1,0)$, $\varepsilon = 3$. The attacker sets $z = (?,?,?,?,?,?,?)$. Session~1: The client authenticates with $y = (1,1,0,1,0,1,0)$. In this case, $d(x,y) = 3 \leq \varepsilon$. The values of the erroneous bits of the client are obtained, yielding $z = (0,0,1,?,?,?,?)$. Session~2: the client authenticates with $y = (0,0,0,0,1,1,0)$. In this case, $d(x,y) = 3 \leq \varepsilon$, and the attacker obtains the value of the erroneous bits of the client and updates $z = (0,0,1,1,0,?,?)$. 
    At this point, replacing the unknown values with random bits gives a vector that lies inside the acceptance ball as the number of unknown coordinates is smaller than the threshold $\varepsilon$.
\end{ex}
In biometrics, some errors happen more frequently than others. In this setup, the Weighted Coupon Collector's Problem must be considered. Each coupon (\textit{i.e.}, each error) has a probability $p_i$ to occur. Suppose that $p_1\leq p_2\leq \dots \leq p_n$ and $\sum_{i=1}^n p_i \leq 1$ then, according to Berenbrink and Sauerwald~\cite{Berenbrink2009Weighted} (Lemma 3.2), the expected number of round $E$ is such that:
\begin{equation}
   \frac{1}{p_1} \leq E \leq \frac{H(n)}{p_1}
\end{equation}
with $H(n)$ the $n$-th harmonic number. The upper bound on $H(n)$ is $1 + \log n$, which yields the expected number of rounds required to complete the collection:
\begin{equation}
   \frac{1}{p_1} \leq E \leq \frac{\ln(n)+1}{p_1}.
\end{equation}
However, while in the original problem one coupon is obtained at each round, the number of errors made by a client during an authentication session is variable, \textit{i.e.},  between $1$ and~$\varepsilon$. In this case, the expected number of rounds required before all the errors have been observed is smaller than in the case where only one error occurs at each round. Consequently, the expected number of rounds required to collect all the errors is still in $\mathcal{O}(\log{n}/p_1)$.

\begin{theorem}\label{th-acc}
    Given $\varepsilon$ a threshold, $x\in\mathbb{Z}_2^n$ a vector, $\texttt{Match}_{x,\varepsilon}$ leaks the positions of the errors and their values below the threshold, and assuming that the rarest coupon is obtained with probability $p_1=n^{-\alpha}$ with $\alpha\in\mathbb{R}_{\ge 1}$ an attacker can retrieve $x$ in $\mathcal{O}(n^{\alpha}\log{n})$.
\end{theorem}

\begin{proof}
According to the Weighted Coupon Collector's problem and assuming that the rarest coupon is obtained with probability $p_1=n^{-\alpha}$ with $\alpha\in\mathbb{R}_{\ge 1}$, the vector $x$ is recovered in $\mathcal{O}\left(n^{\alpha}\log{n}\right)$ observations.\hfill $\blacksquare$
\end{proof}

It is worth noting that in this scenario, the attacker does not control the error. If the attacker controls the error locations, then it is possible to obtain $x$ in $\lceil n/\varepsilon \rceil$ queries. This can happen during a fault attack, akin to side-channel attacks. It should also be noted that some coordinates of a biometric data may be non-variable and, as a consequence, an attacker cannot recover them. This partial recovery attack is therefore a privacy attack, and leads to an authentication attack if the number of variable coordinates is sufficiently large (at least $n - \varepsilon$ in the binary case).

\begin{rem}
    In the non-binary case, the value $x_i-y_i$ does not provide enough information. The exact value of $x_i$ can be determined in two cases. First, if $x_i-y_i=-q+1$, then $x_i=0$. Second, if $x_i-y_i=2(q-1)$, then $x_i=q-1$. For all other cases, there is an ambiguity regarding the value of $x_i$ as $y_i$ is unknown. However, by knowing the distribution of $x_i$ and $y_i$, repeating observations yields a statistical attack.
\end{rem}

Attacks for each type of leakage along with their complexities are summarized in Figure~\ref{Leak_table}. 

\section{Concluding Remarks}
\label{Conclu}
Our investigation into the information leakage of a biometric system using privacy-preserving distance has shed light on critical security vulnerabilities that arise under various scenarios. By evaluating the impact of different types of leakage, including distance, error position, and error value, we have examined the potential risks posed to data privacy and security in practical applications. 

Our analysis encompasses `below the threshold' and `below and above the threshold' setups, allowing us to identify specific conditions under which information leakage can have a substantial effect on the overall security of the system. 

It is important to highlight that the leakage `below the threshold' does not significantly hurt the security of the system, while the leakage of `both below and above the threshold' significantly decreases the security. Indeed, the attacks exploiting the leakage `below the threshold' are primarily exponential, while those exploiting information `below and above the threshold' are mainly constant.

The accumulation attack we investigated is based on the assumption of errors uniformly distributed throughout each authentication session. 
The result of the accumulation attack could be further refined by considering a variable number of coupons, randomly drawn between $0$ and $\varepsilon$ in each round, while acknowledging the actual distribution of the errors. As far as we are aware, no previous studies provide an analysis of the distribution of the errors for any systems.
In practical scenarios, certain errors may occur more frequently than others, while some may never occur. A skewed distribution of errors substantially increases the expected number of authentications required from the legitimate user so that the server recovers the hidden template. 
Further research involves refining the accumulation attack as suggested above and exploring other distance metrics, such as $L_1$ (Manhattan) and $L_2$.


\bibliographystyle{unsrt}
\bibliography{biblio}

\end{document}